\documentclass[12pt]{amsart}
\usepackage{amsmath}
\usepackage{amsfonts}
\usepackage{hyperref}
\usepackage{cleveref}
\RequirePackage{amsthm,amsmath,natbib}
\usepackage{graphicx}
\usepackage{amssymb}
\usepackage{fullpage}
\usepackage{enumitem}
\begin{document}
%\setlength{\textwidth}{6.50in}
%\arxiv{math.PR/0000000}
\newtheorem{theorem}{Theorem}
\newtheorem{acknowledgement}[theorem]{Acknowledgement}
\newtheorem{algorithm}[theorem]{Algorithm}
\newtheorem{axiom}[theorem]{Axiom}
\newtheorem{case}[theorem]{Case}
\newtheorem{claim}[theorem]{Claim}
\newtheorem{conclusion}[theorem]{Conclusion}
\newtheorem{condition}[theorem]{Condition}
\newtheorem{conjecture}[theorem]{Conjecture}
\newtheorem{corollary}[theorem]{Corollary}
\newtheorem{criterion}[theorem]{Criterion}
\newtheorem{definition}[theorem]{Definition}
\newtheorem{example}[theorem]{Example}
\newtheorem{exercise}[theorem]{Exercise}
\newtheorem{lemma}[theorem]{Lemma}
\newtheorem{notation}[theorem]{Notation}
\newtheorem{problem}[theorem]{Problem}
\newtheorem{proposition}[theorem]{Proposition}
\newtheorem{remark}[theorem]{Remark}
\newtheorem{solution}[theorem]{Solution}
\newtheorem{summary}[theorem]{Summary}
\newcommand{\comment}[1]{} % this command lets us comment out blocks of text that we're working on.

\title{Non-hedgeable  risk and credit risk pricing}
\author{Juan Dong, Lyudmila Korobenko\and A. Deniz Sezer}
\thanks{Sezer is supported by a discovery grant from the Natural Sciences and Engineering Research Council of Canada}
\address{TransAlta, University of Pennsylvania, University of Calgary}
\email{adsezer@ucalgary.ca}
\date{\today}
\begin{abstract}

We introduce a new model for pricing of corporate bonds, which is a modification of the classical model of Merton.  In this new model, we drop the liquidity assumption of the firm's asset value process, and assume that there is an asset in the market which is correlated with the firm's asset value, and all portfolios can be constructed using solely this asset and the money market account.   We formulate the market price of the corporate bond as the product of the price of the optimal replicating portfolio and $\exp(- \kappa \times \mbox{replication error})$, where $\kappa$ is a positive constant.  The interpretation is that the representative investor accepts the price of the optimal replicating portfolio as a benchmark, however, requests compensation for the non-hedgeable  risk.  We show that if the replication error is measured relative to the firm's value, the resulting formula is arbitrage free with mild restrictions on the parameters.
\end{abstract}
\maketitle

\section{Introduction}

In the classical model of Merton \cite{Merton}, a corporate bond is a contingent claim on the assets of a firm.  A geometric Brownian motion $(V_t)_{t\geq 0}$ models the firm's assets,
\begin{eqnarray}
dV_t=\mu V_tdt+\sigma V_tdW_t \label{eq:value}
\end{eqnarray}
where $\mu\in \mathbb{R}$, $\sigma\in \mathbb{R}^{+}$, and $(W_t)_{t\geq 0}$ is a Brownian motion on some probability space $(\Omega,\mathcal{F},\mathbb{P})$.  The market is endowed with a money market account accumulating interest at a constant rate $r$.
A Merton style bond with face value $D$ and maturity $T$ is the pay-off $\min(V_T,D)$.  The model assumes that the firm's assets are liquidly traded in the market.  This assumption makes the pay-off hedgeable with the firms assets and the money market account. Hence, the arbitrage-free price, $B_t$, of the bond at time $t$ is given by a variation of the Black-Scholes formula:
\begin{eqnarray}
B_t= V_t N(d_1)-D \exp (-r(T-t))N(d_2)
\end{eqnarray}
where $N$ is the standard normal distribution function, $d_1=\frac{ln (V_t/D) +r(T-t)+\frac{1}{2} \sigma^2 (T-t)}{\sigma\sqrt{T-t}}$ and $d_2=d_1-\sigma\sqrt{T-t}$. 

%An interesting feature of this formula is that the price of the bond does not depend on the drift $\mu$ of the firm's assets, which raises a possible viewpoint for pricing of default risk.  Consider two firms, Firm A and Firm B, which both issue zero coupon bonds at time $0$ with the same face value and maturity. Suppose that Firm A has a higher probability of default than Firm B.  One would intuitively think that the price of the bond issued by Firm A must be lower than the bond issued by Firm B, because the investors would want to be compensated for taking the extra risk.  Merton's model gives a counter example to this reasoning.  If Firm A and Firm B have the same asset values at time $0$, the same volatility, but different drifts with $\mu_B>\mu_A$. Then, the probability of default, $\mathbb{P}(V_T<D)$, is higher for Firm A than Firm B. However, the prices of the bonds issued by the two firms are the same. In Merton's setting, this is natural, because both  bonds are hedgeable, and the prices of the replicating portfolios are the same since the firm values are the same. 

%Perfect replication of the pay-off $\min(V_T,D)$ eliminates the effect of the drift term in the pricing of the bond.  

Note that, in Merton's model, perfect replication is possible because of the unrealistic assumption that the assets of the company are liquidly traded in the market. In this paper, we follow Merton's model, but assume that the firm's assets are not liquidly traded.  Instead, there is an asset in the market, $S_t$, that is correlated with $V_t$, and all portfolios can be constructed using solely this asset and the money market account. As now the model is incomplete, there is no unique way to set the price of the claim $\min(V_T,D)$. 

Similar modifications of Merton's model have been considered in the literature by various authors (see e.g. \cite{Leung}, \cite{Bielecki}, \cite{Okhra}), who have used techniques such as indifference pricing, mean variance hedging, local risk minimization  for the valuation of the claim $\min(V_T,D)$.  We propose a different valuation approach that can be considered as an extension of mean variance hedging. We model the price process $(B_t)_{t\geq 0}$ of the claim $\min(V_T,D)$ as a stochastic process adapted to the natural filtration of $S_t$ and $V_t$. We write $B_t$ as the product of two components.  One component is the price of the optimal replicating portfolio for $\min(V_T,D)$ determined at time $t$; that is, the optimal replicating portfolio is solved for the time horizon $[t,T]$ at each instant $t$.  The other component is a discount factor that takes into account the replication error. More specifically, 
\begin{eqnarray}\label{pricingrule}
B_t& =& \mbox{The price of the time $t$-optimal replicating portfolio}\\
&&\times e^{-\kappa \times \mbox{the time $t$-minimum replication error}},\nonumber
\end{eqnarray}
where $\kappa$ is a positive constant.  

The interpretation is that the buyer (or the seller) of the bond at a given time $t$ uses the price of the optimal replicating portfolio as a benchmark, however, requests (or, respectively offers) compensation for the additional risk involved in the claim $\min(V_T,D)$ (more on this later).   The optimality criterion for the replicating portfolio is the expected square difference between the terminal value of the portfolio and the claim $\min(V_T,D)$.  This optimization problem is referred to in the literature as the mean variance  hedging (MVH) problem \cite{Schweizer}.  Note that we consider the dynamic version of the mean variance hedging problem; since at each time $t$, we re-solve the optimal replication problem for the period $[t,T]$ conditional on the information available up to time $t$.  Let $b_t$ be the cost of the optimal portfolio found at $t$.  $b_t$ has been proposed as an approximate price for non-hedgeable claims (see e.g. \cite{Schweizer2}).  This is justified because $b_t$, under certain assumptions on the underlying filtration, can be represented as a conditional expectation of the claim with respect to the variance optimal martingale measure equivalent to $\mathbb{P}$ \cite{Bobra}.  This ensures that $(b_t)_{t\in[0,T]}$ forms an arbitrage free price process for the corresponding claim.  An interesting question is whether $(B_t)_{t\in[0,T]}$ in formula \eqref{pricingrule} is arbitrage free as well.

Here we explain the motivation behind the discount factor.  For simplicity, let us assume that $t=0$.  First, we observe that the mean terminal wealth of the optimal replicating portfolio is the same as the mean of the claim $\min(V_T,D)$.  Second, the variance of the terminal wealth of the optimal replicating portfolio is strictly less than the variance of  the claim $\min(V_T,D)$. These properties simply follow from the fact that the optimal replicating portfolio is the projection of the claim $\min(V_T,D)$ on the linear space of random variables of the form $x+\int_0^T \phi_t dS_t$, where $x\in \mathbb{R}$ and $\phi_t$ adapted to the filtration generated by $V$ and $S$.  Let $\Phi_T=x^*+\int_0^T \phi_t dS_t$ be this projection. Then,
\begin{eqnarray*}
\mathbb{E}\left[(\min(V_T,D)- \Phi_T)^2\right]& =& \mbox{Var}\left[\min(V_T,D)\right]+ \mbox{Var}\left[\Phi_T\right]\\
&& -2 \mbox{Cov}\left[ \min(V_T,D), \Phi_T\right]+ \left[\mathbb{E}\min(V_T,D)- \mathbb{E}\Phi_T\right]^2
\end{eqnarray*}
This decomposition implies that $\mathbb{E} \left[\min(V_T,D)\right]= \mathbb{E}\left[\Phi_T\right]$ as the other terms in the right side of the above equation are invariant under shifting the variable $\Phi_T$ by a constant. Since $\min(V_T,D)-\Phi_T$ is orthogonal to $\Phi_T-\mathbb{E}(\min(V_T,D))=\Phi_T-\mathbb{E}(\Phi_T)$, we also have that
\begin{eqnarray}
\nonumber\mbox{Var}\left[\min(V_T,D)\right]& =& \mathbb{E}\left[ (\min(V_T,D)- \Phi_T)^2\right]+\mathbb{E}\left[(\min(V_T,D)- \Phi_T) (\Phi_T-\mathbb{E}(\min(V_T,D))\right]\\
&&+ \mbox{Var}\left[\Phi_T\right]\nonumber \\
&=&\mathbb{E}\left[(\min(V_T,D)- \Phi_T)^2\right]+\mbox{Var}\left[\Phi_T\right] \label{dec}.
\end{eqnarray}
Hence, $\mbox{Var}\left[\min(V_T,D)\right]> \mbox{Var}\left[\Phi_T\right]$. The existence of a discount factor for the price of the claim $\min(V_T,D)$ is consistent with the Markowitz portfolio choice theory: the price of the claim $\min(V_T,D)$ must be less than the price of the optimal replicating portfolio.  It is intuitive that the discount factor should be a function of the replication error, as the replication error is exactly the difference between the variance of the claim $\min(V_T,D)$ and the variance of the terminal wealth of the optimal replicating portfolio as the equation \eqref{dec} shows.  Replacing the expectations with conditional expectations, similar arguments hold for an arbitrary time $t$, as well.  The choice of the exponential function is for convenience: we need the discount factor to be non-negative, and less than $1$, and converge to $1$ as $t\rightarrow \infty$. Any other function satisfying these properties can be used as well, but we prefer the exponential function for mathematical convenience.

Dynamic MVH can be formulated as a stochastic control problem as in \cite{Bertsimas}, \cite{Bobra}, \cite{JeanBlanc}, \cite{Kohlmann2}.  Corresponding stochastic control problem can be solved using the Hamilton Jacobi Belman equation if the underlying processes are Ito processes (as in \cite{Bertsimas}), or by backward stochastic differential equations if the underlying processes are general semi-martingales (as in \cite{Bobra}, \cite{Kohlmann2}, \cite{JeanBlanc}).  Since our underlying processes are geometric Brownian motions, we rely on the results of \cite{Bertsimas} to obtain an analytical formula for the price of $\min(V_T,D)$.  Our main contribution is that under certain conditions on the parameters, the price process $(B_t)_{0\leq t\leq T}$ in formula \eqref{pricingrule} satisfies the NFLVR (no free lunch with vanishing risk) condition of \cite{DalSch}.

This paper in its origin is tied to the credit risk premium puzzle, which refers to the inability of the classical structural models, such as Merton's model, to predict credit spreads of corporate bonds of short maturity \cite{CCG}.  In particular, it is well known that Merton's model, when calibrated to the historical default rates of a firm, significantly underestimates the short term credit spreads of the corporate bonds of the same firm. It is natural to attribute this lack of fit to the unrealistic assumptions of the model:  (1)  The firm's asset value is perfectly observed; (2)  the firm's assets are liquidly traded.     The imperfect information models \cite{DL}, \cite{GJZ},\cite{JPS}) are variations of Merton's model relaxing the perfect information assumption, and thereby allowing the default time to have a conditional hazard rates as in reduced form models to give a better fit to short term credit spreads.   However, the incompleteness of the information makes the model incomplete. Hence, there is a need to model the market's preferences for risk and return as well, which is not taken into consideration in the information reduction papers so far.  Here, we acknowledge that the current paper does not provide a full story either; eventhough we incorporate preferences into the price.  We note that the assumptions (1) and (2) can be seen as the two sides of the same coin; it is hard to imagine that one would hold in the absence of the other. For simplicity, in this paper, we unrealistically assume that (1) is true, wheres (2) is false. Inevitably, our short term spreads resemble the short term spreads of a perfect information model.  Hence, our results should not be considered as an attempt to solve the credit risk premium puzzle, but as a preliminary analysis to a more powerful methodology which relaxes the assumption (1) as well.     

It is also worthwhile to compare our approach to indifference pricing.   Indifference price $p$ of a defaultable bond is the price defined for an initial endowment $v$ that makes the buyer indifferent between investing the $v$ dollars in the default-free  market versus buying the defaultable bond and investing the remaining $v-p$ dollars in the default-free market (see e.g. \cite{Bielecki}, \cite{Leung}). The comparison of the investment strategies are done via comparing the expected utility of the final wealth, with respect to a pre-specified utility function.  Since this formulation is done from the buyer's point of view, this is also called the buyer's indifference price.  A similar formulation can be made from the point of view of the seller, which yields the seller's indifference price, which is in general different from the buyer's price.  The main difference of our approach from indifference pricing is that the price in our model is not characterized as a break even point for an optimal investment problem either from an individual buyer's or seller's point of view. The formula \eqref{pricingrule} represents the market price of the claim $\min(V_T,D)$, and it is determined by a representative agent (representing the aggregate behavior of the investors in the economy). The interpretation is that the market settles on a price so that the risk-return profile of the product is compatible with the other instruments in the market.  The discount factor represents the market's risk aversion towards a product that is more risky.  The key question for us is whether this model is a legitimate model for the market price of the claim $\min(V_T,D)$, hence the question of arbitrage.  Whereas, in indifference pricing, the question of arbitrage is not the main question of interest since the indifference price represents an individual's own valuation of the claim $\min(V_T,D)$, regardless of its market price.   

Our pricing framework is amenable to extensions to more general credit risk derivatives and/or to multifactor models due to existence of a vast literature on the mean variance hedging problem and its solution in terms of backward stochastic differential equations.  In particular, both the price of the optimal replicating portfolio $(b_t)_{t\geq 0}$ and the replication error $(c_t)_{t\geq 0}$ can be represented as solutions of BSDEs. (see e.g \cite{Bobra})  These representations can be explored to answer the no arbitrage property in a more general framework.

Organization of the paper is as follows:  In section \ref{sec:scp}, we define our setting and review the MVH problem for a general contingent claim and the corresponding stochastic control problem.  In particular, we review \cite{Bertsimas}'s characterization of the solution of the stochastic control problem in terms of a certain boundary value problem for a system of partial differential equations(PDEs). In section \ref{sec:pricing}, we introduce our pricing model for a general contingent claim and apply the results of \cite{Bertsimas} to formulate the price in terms of the solutions of a certain system of PDEs.  In section \ref{sec:bondpricing}, we focus on the pricing of the contingent claim $\min(D,V_T)$ and obtain an analytical formula for the price by solving the corresponding system of PDEs. We then use these explicit formulas to analyze the qualitative properties of the price, in particular, whether it is arbitrage free.  In section \ref{sec:numerical} we examine and perform numerical experiments to see how the price and yield spreads react to changes in the parameters.  Appendix contains the derivation of the solutions to the PDEs.  

\section{MVH as a stochastic control problem} \label{sec:scp}

In this section, we define the MVH problem and give an overview of the stochastic control approach of \cite{Bertsimas} to solve the MVH.  Let $V_t$ be as defined in the formula \eqref{eq:value}.  From now on, we denote the drift of $V_t$ as $\mu_1$, the volatility of $V_t$ as $\sigma_1$ and the Brownian motion driving $V$ as $W^1$. Let $S_t$ be another geometric Brownian motion:
\begin{eqnarray}
dS_t=\mu_2 S_tdt+\sigma_2 S_tdW^2_t \label{eq:asset},
\end{eqnarray}
where $W_t^1$ and $W_t^2$ are two Brownian motions with correlation $d W_t^1dW_t^2 = \rho dt$.

We assume a trading horizon $[0,T]$, $T>0$, and that initially the market contains only a traded asset with the price process $(S_t)_{t\geq 0}$ and a money market account which accumulates interest at a constant deterministic rate $r>0$.  In order not to introduce new notation, we assume that both $V_t$ and $S_t$ are already discounted.  We consider including in this market contingent claims with discounted pay-offs at $T$ formulated as functions $F(S_T,V_T)$ of $S_T$ and $V_T$. 

The dynamic for the discounted wealth of a self-financing portfolio is
\begin{equation*}
  dP_t = \theta_tdS_t = \theta_t \mu_{2}S_tdt+\theta_t\sigma_{2} S_tdW_t^2.
\end{equation*}
where $\theta_t$ is a predictable $S$-integrable process with respect $\mathcal{F}_t=\sigma(V_s,S_s, s\geq 0)$.
Since there are more random factors than the traded assets in the market, the market is incomplete. Therefore, there does not exist a replicating portfolio that perfectly hedges every given contingent claim $F(S_T,V_T)$. MVH refers to finding an optimal trading strategy that best approximates the payoff $F(S_T, V_T)$; that is,
\begin{center}
  \emph{minimize $\mathbb{E}[(P_T-F(S_T, Y_T))^2]$  $\qquad$ over all pairs ($p, \theta$)}
\end{center}
where the pair ($p, \theta$) describes a dynamic trading strategy which starts at time $0$ with initial capital $p$ and holds $\theta_t$ shares of the traded asset at time $t$ and is self-financing, thus leading to a wealth of $P_t$ at time $t$. Here $P_t$ is the discounted wealth of the portfolio.  This problem can be treated as a stochastic optimal control problem if we rewrite it as
\begin{center}
  \emph{\mbox{minimize} $\mathbb{E}[(P_T-F(S_T, V_T))^2]$  $\qquad$ over all $\theta \in \Theta$}
\end{center}
\emph{with the dynamics}
\begin{eqnarray*}
&& dV_t = \mu_{1} V_tdt+\sigma_{1} V_tdW_t^1,\\
&& dS_t = \mu_{2}S_tdt+\sigma_{S} S_tdW_t^2,\\
&& dP_t = \theta_tdS_t = \theta_t\mu_{2}S_tdt+\theta_t\sigma_{2} S_tdW_t^2,\\
&& P_0 = p, S_0=s, V_0=v.
\end{eqnarray*}
where $\Theta$ is the set of all $\mathbb{R}$-valued predictable $S$-integrable processes such that $\int_0^T \theta_t dS_t$ is well-defined and square integrable.  In this context, the 3-dimensional process $(P_t, S_t, V_t)$ is the state process, the process $(\theta_t)_{t\in[0,	T]}$ is the control process, and $\Theta$ is the constraint set. 

Note that in the above formulation, the initial cost of the portfolio $p$ is taken as fixed. Once the above problem is solved for any $p$, then one can optimize over $p$. To make this clearer, we firstly solve the optimal trading strategy $\theta^*(p,s,v)$ for this given $p$. Then we define
\begin{equation*}
  \varepsilon(s,v) := \sqrt{\min_{p} \mathbb{E}_{p,s,v}\left[(P_T-F(S_T, V_T))^2\right]}.
\end{equation*}
Let $p^{\star}(s,v)$ be the optimal solution of the above problem.  The uniqueness of $p^{\star}(s,v)$ is well known; in fact, $\mathbb{E}_{p,s,v}\left[(P_T-F(S_T, V_T))^2\right]$ is a quadratic function of $p$ \cite{Bertsimas}.   The optimal replicating strategy $\theta^{\star}(s,v)$ is $\theta(p^{\star}(s,v),s,v)$ and the minimal replication error is $\varepsilon(s,v)$.  Note that above we did not require $\Theta$ to contain only admissible processes (i.e. $\theta_t$ such that $P_t\geq 0$ for all $t$).  However, it turns out that  the optimal replicating strategy $\theta^{\star}(s,v)$ is always admissible .  This follows from the fact that $S$ is continuous and therefore the value process of the optimal replicating portfolio can be represented as a conditional expectation of the final pay-off with respect to the variance optimal martingale measure (see e.g.\cite{Bobra}).

Using the Markov property of $(S_t,V_t)$, \cite{Bertsimas} represents the control process $\theta$ as $\theta_t = \theta(t, P_t, S_t, V_t)$ .   Let $\mathbb{E}_{t,p,s,v}(\cdot)$ be the conditional expectation operator $\mathbb{E}(\cdot|P_t=p, S_t=s, V_t=v)$. Let $\Theta_t=\{(\theta_s)_{s\in [t,T]}: \int_t^T \theta_sdS_s \mbox{ is well defined  and square integrable}\}$.  As usual in stochastic control, we consider the dynamic version:
\begin{center}
  $\mbox{minimize}$ $\mathbb{E}_{t,p,s,v}[(P_T-F(S_T, V_T))^2]$  $\qquad$ over all $\theta \in \Theta_t$.
\end{center}
%$P_t = p, S_t=s, V_t=v$, and for $s>t$,
%\begin{eqnarray*}
%&& dV_s = \mu_{1} V_sds+\sigma_{1} V_sdW_s^1,\\
%&& dS_s = \mu_{2} S_sds+\sigma_{2} S_sdW_s^2,\\
%&& dP_s = \theta_sdS_s = \theta_s\mu_2 S_sds+\theta_s\sigma_{2} S_sdW_s^2.
%\end{eqnarray*}
Let $V(t,p,s,v)$ be the optimal value function of this control problem.  It is well known that $V(t,p,s,v)$ is characterized as the solution of the Hamilton Jacobi and Bellman equation:

\begin{eqnarray}\label{HJB} \lefteqn{\frac{\partial V}{\partial t}+\inf_{\{\theta_t\}}\{[\mu_{2} s \frac{\partial}{\partial s}+\mu_1 v\frac{\partial}{\partial v}+\theta_t \mu_2 s\frac{\partial}{\partial p}+\frac{1}{2}\sigma_{2}^2s^2\frac{\partial^2}{\partial s^2}+\frac{1}{2}\sigma_{1}^2v^2\frac{\partial^2}{\partial v^2}}\\
&&+\sigma_{2} \sigma_{1} sv\rho\frac{\partial}{\partial s\partial v}+\frac{1}{2}\theta_t^2\sigma_{2}^2s^2\frac{\partial^2}{\partial p^2}+\sigma_{2}^2s^2\theta_t\frac{\partial^2}{\partial s\partial p}+\sigma_{2}\sigma_{1}sv\theta_t\rho\frac{\partial^2}{\partial v\partial p}]V\}\nonumber\\
&&= 0 \nonumber,
\end{eqnarray}
with the boundary condition
\begin{eqnarray*}
  V(T, p, s, v) = [p-F(s, v)]^2.
\end{eqnarray*}
The optimal choice of $\theta_t$, denoted by $\theta^*(t,V_t,S_t,P_t)$  satisfies:
\begin{eqnarray*}
\theta^{\star}_t & =& \frac{-\mu_2 s\frac{\partial V}{\partial p}-\sigma_{2}^2 s^2 \frac{\partial^2 V}{\partial s \partial p}-\rho\sigma_{1}\sigma_{2t} vs\frac{\partial^2 V}{\partial v \partial p}}{\sigma_{2}^2 s^2 \frac{\partial^2 v}{\partial p^2}}.
\end{eqnarray*}
Substituting the expression for $\theta^{\star}_t$ into the PDE \eqref{HJB} gives us a PDE for the only unknown function $V$.

The hard work of dynamic programming is solving the highly nonlinear PDE for $V$.  For MVH, this process is facilitated by the observation that the value process of the MVH problem has a quadratic structure.  More specifically, \cite{Bertsimas} proves the following theorem:

\begin{theorem}\cite{Bertsimas} The value function $V(t, p, s, v)$ is quadratic in the wealth of the portfolio $p$, i.e. there are continuous functions $a(t, s, v)$, $b(t, s, v)$ and $c(t, s, v)$ such that
\begin{equation*}
  V(t, p, s, v) = a(t, s,v )\cdot[p-b(t, s, v)]^2 + c(t, s, v), \qquad 0 \leq t \leq T,
\end{equation*}
where $a$, $b$ and $c$ satisfy
\begin{eqnarray}
\frac{\partial a}{\partial t} &=& (\frac{\mu_2}{\sigma_2})^2a+\mu_{2} s \frac{\partial a}{\partial s}+[\frac{2\sigma_{1}\rho v\mu_2}{\sigma_{2}}-\mu_{1}v]\frac{\partial a}{\partial v}-\frac{1}{2}v^2s^2\frac{\partial^2 a}{\partial s^2}-\frac{1}{2}\sigma_{1}^2v^2\frac{\partial^2a}{\partial v^2}\nonumber\\
&& -\: v^2\sigma_{1} s\rho\frac{\partial^2 a}{\partial s \partial v}+\frac{1}{a}v^2s^2(\frac{\partial a}{\partial s})^2+\frac{1}{a}\rho^2 \sigma_{1}^2v^2(\frac{\partial a}{\partial v})^2+2v^2\sigma_{1} s\rho\frac{\partial a}{\partial s}\frac{\partial a}{\partial v}, \label{eq_a}\\
\frac{\partial b}{\partial t} &=& [\frac{\sigma_{1}}{\sigma_{2}}v\rho\mu_{2}-\mu_{1}v]\frac{\partial b}{\partial v}-\frac{1}{2}\sigma_{1}^2v^2\frac{\partial^2b}{\partial v^2}-\frac{1}{2}\sigma_{2}^2s^2\frac{\partial^2b}{\partial s^2}-\sigma_{2}\sigma_{1}sv\rho\frac{\partial^2b}{\partial v \partial s}\nonumber\\
&& +\: \frac{\sigma_{1}^2v^2}{a}(\rho^2-1)\frac{\partial a}{\partial v}\frac{\partial b}{\partial v},\label{eq_b} \\
\frac{\partial c}{\partial t} &=& -(\mu_1 v\frac{\partial c}{\partial v}-\mu_{2}s\frac{\partial c}{\partial s}-\frac{1}{2}\sigma_{1}^2v^2\frac{\partial^2c}{\partial v^2}-\sigma_{2}\sigma_{1}sv\rho\frac{\partial^2c}{\partial s \partial v}-\frac{1}{2}\sigma_{2}^2s^2\frac{\partial^2c}{\partial s^2}\nonumber\\
&& +\: a\sigma_{1}^2v^2(\rho^2-1)(\frac{\partial b}{\partial v})^2\label{eq_c},
\end{eqnarray}
with boundary conditions
\begin{equation*}
  a(T, s, v) = 1,
\end{equation*}
\begin{equation*}
  b(T, s, v) = F(s, v),
\end{equation*}
\begin{equation*}
  c(T, s, v) = 0.
\end{equation*}
\end{theorem}

Given the quadratic structure of the value process, we can write $V_t$ as a function of $p$:
\begin{equation*}
 V(t,p,s,v) = a(t, s,v)\cdot[p-b(t, s, v)]^2+c(t, s, v).
\end{equation*}
\cite{Bertsimas} shows that $a(t, s, v)$ is positive, hence the initial wealth $p^{\star}(t,s,v)$ that minimizes the quadratic function is exactly $b(t, s,v)$. The optimal-replication strategy is the $\theta$ corresponding to this initial wealth $p^{\star}(t,s,v)$, and the minimum replication error over all $p$ is $\varepsilon^{\star} (t,s,v)=\min_p \sqrt{V(t,p,s,v)} = \sqrt{c(t, s, v)}$. 

\section{A Pricing Model} \label{sec:pricing}
In this section, we propose a pricing model to price a contingent claim of the form $F(V_T)$.  Because the pay-off is only a function of $V_T$ and $T$ but not $S_T$, it turns out that the functions $a$, $b$ and $c$ are also only functions of $t$ and $v$ but not $s$.  (See appendix for a justification.)  

We propose the following formula to calculate the price of a contingent claim $F(V_T)$: 
\begin{equation}\label{proposed}
 B(t, V_t) = b(t, V_t)\cdot e^{-\kappa \cdot \frac{c(t, V_t)}{d(t,V_t)}},
\end{equation}
where $b(t, V_t)$ is the cost of building an optimal-replicating portfolio for the claim $F(V_T)$ at time $t$, and $c(t, V_t)$ is the expected squared error of the replication.  $d(t,V_t)$ is a suitable normalization factor; it is not specified at the moment, but we assume $d(t,v)$ to be strictly positive, first order differentiable in $t$ and second order differentiable in $v$, and
\begin{equation}\label{limd} 
\lim_{t\rightarrow T} \frac{c(t, v)}{d(t,v)}=0.
\end{equation} 
Here let us comment on the main features of formula \eqref{proposed}. The price of the contingent claim $B(t, V_t)$ converges to the payoff $F( V_T)$ as $t$ approaches $T$. This is due to boundary condition of $b(t, v)$ and the assumption \eqref{limd}.  Recall $b(t,V_t)\rightarrow F(V_T)$ as $t\rightarrow T$.  Also, because $c(t,v)\rightarrow 0$ as $t\rightarrow T$, the assumption \eqref{limd} will be satisfied for a wide range of choices for $d$.   We interpret the term $b(t, V_t)$ as a benchmark price.  Indeed, in a sense, it is the price of the closest traded instrument available in the market. However, the variability of this benchmark instrument is less than the variability of the contingent claim, therefore it makes sense for the investor to ask for a discount from the benchmark price; which is the rationale for the term $e^{-\kappa \cdot \frac{c(t, V_t)}{d(t,v)}}$.  Here $\kappa$ is a preference parameter, a higher value indicates a higher level of risk aversion.  Note that the mean squared replication error $c(t,V_t)$ has been normalized by $d(t,V_t)$, which allows the investor to measure the replication error in relative terms. A possible choice for $d(t,v)$ is the conditional variance of $F(V_T)$ given $S_t=x,V_t=v$. Simpler normalization factors can be chosen depending on the type of the contingent claim.  For example, when $F(V_T)=\min(D,V_T)$, that is, a Merton style bond, it may make sense to choose $d(t,v)=v^2$. That is, the replication error is measured relative to the value of the firm.  This not only gives a better performance measure for the replication (for example $c(t,V_t)=800$ would be more alarming if the firm value were $10$, as compared to $100$), but also makes $\kappa$ a unitless constant.  In general, we need normalization for technical reasons, in particular, to show that the pricing formula is arbitrage free.        

In the rest of this section, we investigate whether the proposed pricing model is arbitrage free.  We use the result of Dalbean and Schachermayer \cite{DalSch} which states that there is no arbitrage in the sense of \textit{no free lunch with vanishing risk} if and only if there exists an equivalent probability measure $\mathbb{Q}$ rendering the price processes local martingales. Here  note that $B_t$ and $S_t$ are locally bounded semimartingales.  Our first goal is to find two processes $\lambda_1(t)$ and $\lambda_2(t)$ such that if $\tilde{W}^i_t=\int_0^t\lambda_i(s)ds+  W^i_t$ then for $t\in [0,T]$,
\begin{eqnarray} 
\label{stcond_1}
B_t&=&B_0+ \int_0^t N_t d\tilde{W}^{1}_t\\ 
S_t&=&S_0+ \int_0^t L_t d\tilde{W}^{2}_t \label{stcond2}
\end{eqnarray}
for some continuous and adapted processes $N_t$ and $L_t$.  Then the next question is when there is a probability measure $\mathbb{Q}$ under which  $\tilde{W}^{1}_t$ and $\tilde{W}^{2}_t$ are martingales.

%\begin{theorem} \label{theo:abscon}  Let $X=(X_t)_{t\in[0,T]}$ be a two-dimensional Brownian motion defined on $\Omega,\mathcal{F}$.  Assume the canonical setting, that is, $\Omega$ is the space of continuous functions from $[0,T]$ to $\mathbb{R}^2$ and $X_t(\omega)=\omega_t$ and $\mathcal{F}$ is generated by $X$ and $P$ is the Wiener measure.   Let $A(t,x)=(A_1(s,x),A_2(s,x))$ be a two dimensional vector of bounded functions on $[0,T]\times \mathbb{R}^2$.  Then there exists unique a probability measure $\mathbb{Q}$ on $(\Omega,\mathcal{F})$ such that 
%$X_t-\int_0^t A(s,X_s) ds$ is a two dimensional Brownian motion.  Moreover, if $\mathbb{P}(\int_0^T A_1^2(s, X_s) ds+ \int_0^T A_2^2(s, X_s) ds <\infty)=\mathbb{Q} (\int_0^t A_1^2(s, X_s) ds+ \int_0^t A_2^2(s, X_s) ds <\infty)=1$ then $\mathbb{P}$ and $\mathbb{Q}$ are equivalent.
%\end{theorem}

%\begin{proof}  The existence and uniqueness of $\mathbb{Q}$ follows from Theorem 2.34, Chapter III of \cite{JS}.  The equivalence of $\mathbb{P}$ and $\mathbb{Q}$ follows from  Theorem 5.34, Chapter III of \cite{JS} for which the required local uniqueness condition follows from Theorem 2.40 and its corollary 2.41, Chapter III of \cite{JS}.
%\end{proof}

\begin{theorem}\label{noarbitrage} 

Let $\tilde{c}(t,v)= \frac{c(t,v)}{d(t,v)}$.  %Assume $\tilde{c}$ is twice differentiable with their first and second order derivatives are continuous on $(0,T)\time(0,\infty)$ %and bounded on $(0,T)\times K$, for any compact $K\subset(0,\infty)$.  
Assume $\frac{\partial b}{\partial v}-\kappa b\frac{\partial \tilde{c}}{\partial v}>0$ for all $(t,v)\in (0,T)\times (0,\infty)$. Let $L_t=S_t\sigma_2$, and $(M_t)_{t\in [0,T]}$ and $(N_t)_{t\in[0,T]}$ be as defined below:  
\begin{eqnarray}
M_t &=&e^{-\kappa \tilde{c}_t} \left[\frac{\partial b}{\partial t}(t,V_t)+\frac{\partial b}{\partial{v}}(t,V_t)\mu_1 V_t+\frac{1}{2}\frac{\partial^2 b}{\partial{v^2}}(t,V_t)\sigma_1^2 V_t^2\right]\label{formulaforM}\\
&&-\kappa b_t e^{-\kappa \tilde{c}_t}\left[\frac{\partial \tilde{c}}{\partial t}(t,V_t)+\frac{\partial \tilde{c}}{\partial{v}}(t,V_t)\mu_1 V_t+\frac{1}{2}\frac{\partial^2 \tilde{c}}{\partial{v^2}}(t,V_t)\sigma_1^2 V_t^2\right]\nonumber\\
&&+  \frac{1}{2} \kappa^2 b_t e^{-\kappa \tilde{c}}\frac{\partial \tilde{c}}{\partial{v}}(t,V_t)^2\sigma_1^2 V_t^2 -\kappa e^{-\kappa \tilde{c}_t} \frac{\partial \tilde{c}}{\partial{v}}(t,V_t)\frac{\partial b}{\partial{v}}(t,V_t)\sigma_1^2 V_t^2,\nonumber\\
\end{eqnarray}
\begin{eqnarray*}
N_t &=& e^{-\kappa \tilde{c}_t} V_t\sigma_{1}\left[\frac{\partial b}{\partial v}(t,V_t)-\kappa b_t\frac{\partial \tilde{c}}{\partial v}(t,V_t)\right].\\
\end{eqnarray*}
Then $N_t>0$ for all $t\in [0,T)$, and equations \eqref{stcond_1} and \eqref{stcond2} hold with $\tilde{W}^i_t=\int_0^t\lambda_i(s) ds+ W^i_t$, where $\lambda_1(t)= \frac{M_t}{N_t}$ and $\lambda_2=\frac{\mu_2}{\sigma_2}$.

Moreover, if 
\begin{equation}
\mathbb{P}(\exp(\frac{1}{2(1-\rho^2)}\int_0^T \left[\frac{1}{2}(\frac{M_t}{N_t})^2+\rho \frac{M_t\mu_2}{N_t\sigma_2}\right]dt) <\infty), \label{novikov}
\end{equation}
then there exists a probability measure $\mathbb{Q}$ equivalent to $\mathbb{P}$ such that $(\tilde{W}^{1}_t,\tilde{W}^{2}_t)_{0\leq t\leq T}$ is a two dimensional Brownian motion with correlation $\rho$.

\end{theorem}

\begin{proof}

We first show that the equations \eqref{stcond_1} and \eqref{stcond2} are true.  Note that
\begin{eqnarray*}
S_0+ \int_0^t L_t d\tilde{W}^{2}_t& =&S_0+\int_0^t S_s\sigma_2 \frac{\mu_2}{\sigma_2}ds+\int_0^t S_s\sigma_2 dW_s^2\\
&=& S_t.
\end{eqnarray*}

Similarly,
\begin{eqnarray*}
B_0+ \int_0^t N_t d\tilde{W}^{1}_t&=& B_0+\int_0^t \frac{M_s}{N_s} N_s ds+\int_0^t N_s dW^1_s\\
&=&B_0+\int_0^t M_s ds+ \int N_s dW^1_s\\
&=&B_0+\int_0^t\left( \frac{\partial}{\partial t} (be^{-\kappa \tilde{c}})(s,V_s)+\frac{1}{2}\frac{\partial^2}{\partial v^2} (b e^{-\kappa \tilde{c}})(s,V_s)\sigma_1^2V_s^2\right)ds\\
&& +  \int_0^t \frac{\partial b}{\partial v}{ e^{-\kappa \tilde{c}}} (s,V_s)dV_s \\
&=&b(t,V_t) e^{-\kappa \tilde{c}(t,V_t)},
\end{eqnarray*}
where the last line follows from the Ito formula.

Note that because $N_t>0$, and both $M_t$ and $N_t$ are continuous, the process 
\[X_t=-\int_0^t \frac{\lambda^1_s-\rho\lambda^2_s}{1-\rho^2} dW^1_s-\int_0^t\frac{\lambda^2_s-\rho\lambda^1_s}{1-\rho^2} dW^2_s\] 
is a continuous local martingale. Let $(Z_t)_{t\in[0,T]}$ be the stochastic exponential of $X_t$.  By the well known Novikov condition (see e.g. \cite{KS}), $(Z_t)_{t\in[0,T]}$ is a  strictly positive martingale if $\mathbb{E}(\exp \frac{1}{2} <X_T,X_T>)<\infty$.  An elementary calculation gives
\begin{eqnarray*}
<X_t,X_t>=\frac{1}{1-\rho^2} \int_0^t \left[(\lambda^1_s)^2+\frac{1}{2}(\lambda^2_s)^2 -2\rho \lambda^1_s\lambda^2_s\right] ds.
\end{eqnarray*}
Hence, we have 
\[\exp(\frac{1}{2}<X_T,X_T>)=\exp(\frac{1}{2(1-\rho^2)}\int_0^T\left[ \left(\frac{M_t}{N_t}\right))^2+ \left(\frac{\mu_2}{\sigma_2}\right)^2-2\rho \frac{M_t\mu_2}{N_t\sigma_2}\right]dt.\]
Hence, the hypothesis \eqref{novikov} implies that the Novikov condition is satisfied. Therefore, $Z_t$ is a martingale with $Z_0=1$. We define the probability measure $\mathbb{Q}$ equivalent to $\mathbb{P}$ as $\mathbb{Q}(A)=\mathbb{P}(Z1_A)$.   Since $dZ_t=Z_t \left(\frac{\lambda^1_s-\rho\lambda^2_s}{1-\rho^2} dW^1_s-\int_0^t\frac{\lambda^2_s-\rho\lambda^1_s}{1-\rho^2} dW^2_s\right)$, for $i=1,2$
\begin{eqnarray*}
d(Z_t\tilde{W}^i_t)&=&Z_td\tilde{W}^i_t+\tilde{W}^i_t dZ_t+d<Z_t,\tilde{W}^i_t>\\
&=&Z_tdW^i_t+\tilde{W}^i_tdZ_t+Z_t\lambda^i_tdt+ Z_t(\frac{\lambda^i_s-\rho\lambda^{i^*}_s}{1-\rho^2} ds+\rho \frac{\lambda^i_s-\rho\lambda^{i}_s}{1-\rho^2})dt\\
&=&Z_tdW^i_t+\tilde{W}^i_tdZ_t
\end{eqnarray*}
where $i^*=2$ if $i=1$, and $i^*=1$ if $i=2$. Hence, both $\tilde{W}^1$ and $\tilde{W}^2$ are $\mathbb{Q}$ local martingales which implies that $(\tilde{W}^1_t, \tilde{W}^2_t)_{t\in[0,T]}$ is a two dimensional Brownian motion with correlation $\rho$ by Levy's characterization theorem.   

\end{proof}

\section{Application to credit risk pricing} \label{sec:bondpricing}
In this section, we apply the results of the previous section to the credit risk pricing model.  The price of the corporate bond is
\begin{equation*}
 B(t, s, v) = b(t, s, v)\cdot e^{-\kappa \cdot \frac{c(t, s, v)}{v^2}},
\end{equation*}
where the parameters $a(t, s, v)$, $b(t, s, v)$ and $c(t, s, v)$ are the solutions to the system of PDEs \eqref{eq_a}, \eqref{eq_b}, \eqref{eq_c},
with boundary conditions
\begin{equation}
  a(T, s, v) = 1, \label{eq:bdca}
\end{equation}
\begin{equation}
  b(T, s, v) = \min(v, D), \label{eq:bdcb}
\end{equation}
\begin{equation}
  c(T, s, v) = 0. \label{eq:bdcc}
\end{equation}

It is possible to solve the above system analytically, and derive explicit formulas for the functions $a$, $b$ and $c$ (see Appendix for details).  
\begin{theorem}  The p.d.e system  \eqref{eq_a}, \eqref{eq_b}, \eqref{eq_c}  with the boundary conditions \eqref{eq:bdca}, \eqref{eq:bdcb}, \eqref{eq:bdcc} have unique solutions $a(t,s,v)$,$b(t,s,v)$ and $c(t,s,v)$ where 

a.  $a(t,s,v)= e^{-\left(\frac{\mu_2}{\sigma_2}\right)^2 {T-t}}$;\\

b. $b(t,s,v)= b(t, v) =  ve^{\left(\mu_1-\frac{\mu_2\rho\sigma_1}{\sigma_2}\right)(T-t)}N(d_1)+D\left(1-N(d_2)\right)$, where
\begin{eqnarray}
d_1(t,v)&=&  \frac{\ln \frac{D}{v}-\sigma_1^2\left[\frac{1}{2}+\frac{1}{\sigma_1^2}(\mu_1-\frac{\sigma_1\rho \mu_2}{\sigma_2})\right](T-t)}{\sigma_1\sqrt{T-t}},\\
d_2(t,v)&=& d_1+ \sigma_1\sqrt{T-t};
\end{eqnarray}

c.  $c(t,s,v)= \sigma^2 (1-\rho^2) \mathbb{E}\left(\int_t^T a_u V_u^2 \left(\frac{\partial b}{\partial v}(u,V_u)\right)^2 du|V_t=v\right)$.

\end{theorem}

\begin{proof}
See Appendix for the analytical derivations of the formulas for $a$, $b$ and $c$.  In particular, all three functions are independent of $s$. Hence we will denote them as $a(t,v)$,$b(t,v)$ and $c(t,v)$.  Below we only prove the representation in (c). 

Applying Ito's formula to function to $c(t,V_t)$, we get
\begin{eqnarray*}
c(T, V_T)-c(t,V_t)&=&\int_t^T (\frac{\partial c}{\partial t}(u, V_u)+\frac{1}{2} \frac{\partial^2 c}{\partial v^2}(u,V_u)) d\langle V_u,V_u \rangle  + \int_{t}^{T}\frac{\partial c}{\partial v} (u, V_u) dV_u\\
&=&\int_t^T (\frac{\partial c}{\partial t}(u, V_u)+ \mu_{1}v\frac{\partial c}{\partial v}(u,V_u)+ \frac{1}{2}\sigma_1^2 v^2 \frac{\partial^2 c}{\partial v^2}(u,V_u))du\\
&& +\int_t^T \sigma_1\frac{\partial c}{\partial v} (u, V_u)V_u dW_u.
\end{eqnarray*}
Since 
\begin{eqnarray*}
\frac{\partial c}{\partial t} &=& -\mu_{1}v\frac{\partial c}{\partial v}-\frac{1}{2}\sigma_{1}^2v^2\frac{\partial^2c}{\partial v^2}+a(t,v)\sigma_{1}^2v^2(\rho^2-1)(\frac{\partial b}{\partial v})^2,
\end{eqnarray*}
we get that 
\begin{eqnarray*}
c(T, V_T)-c(t,V_t)&=&\int_t^T a(u,V_u)\sigma_{1}^2V_u^2(\rho^2-1)(\frac{\partial b}{\partial v}(u,V_u))^2du + \int_t^T \sigma_1\frac{\partial c}{\partial v} (u, V_u)V_u dW_u.
\end{eqnarray*}
Note that $c(T, V_T)=0$.  Hence,   
\begin{eqnarray*}
c(t,V_t)&=& \mathbb{E}\left[\int_t^T a(u,V_u)\sigma_{1}^2V_u^2(1-\rho^2)(\frac{\partial b}{\partial v}(u,V_u))^2du|\mathcal{F}_t\right].
\end{eqnarray*}
Now, the representation in (c) follows from the Markov property of $V_t$.

\end{proof}

\subsection{Analysis of the function $c(t,v)$}

We derive a more explicit formula for the function $c$ and analyze its qualitative properties:

\begin{theorem}  $\frac{c(t,v)}{v^2}$ is monotone decreasing in $v$.  In particular,
\begin{eqnarray}
c(t,v)&=& \sigma_1^2 (1-\rho^2) v^2 e^{[(2\mu_1+\sigma_1^2)](T-t)}\nonumber\\
 &&\int_t^T e^{\left(-(\frac{\mu_2}{\sigma_2})^2-\sigma_1^2-2\frac{\mu_2}{\sigma_2}\rho \sigma_1\right) (T-u)} \mathbb{E}(N(\textbf{d})^2) du\label{formulaforc}
\end{eqnarray}
where $\textbf{d}=\textbf{d}(u,t,v)$ is normally distributed with mean 
\begin{eqnarray}
\mu(t,u,v)&=&\frac{\ln D-\ln v-(\mu_1+\frac{3}{2}\sigma_1^2)(u-t)-\sigma_1^2(\frac{1}{2}+\frac{\alpha}{\sigma_1^2})(T-u)}{\sigma_1 \sqrt{T-u}}\label{eq:mean}
\end{eqnarray}
with variance $\frac{u-t}{T-u}$, where $\alpha=\mu_1-\rho\frac{\mu_2\sigma_1}{\sigma_2}$.  
\end{theorem}

\begin{proof}
We find in the Appendix that
\begin{equation}
\frac{\partial b}{\partial v}(u,v)=  e^{\left(\mu_1-\frac{\mu_2\rho\sigma_1}{\sigma_2}\right)(T-u)}N(d_1(u,v)).
\end{equation}
where
\[d_1(u,v)=  \frac{\ln \frac{D}{v}-\sigma_1^2(\frac{1}{2}+\frac{\alpha}{\sigma_1^2})(T-u)}{\sigma_1\sqrt{T-u}}.\]
Let 
\[Z_u= \frac{V_u^2}{V_0^2} e^{[-2\mu_1-\sigma_1^2]u}\]
Note that
\begin{eqnarray*}
Z_u&=& \frac{V_u^2}{V_0^2} e^{[-2\mu_1-\sigma_1^2]u}\\
&=& e^{[2\mu_1-\sigma_1^2]u+2\sigma_1 W_u} e^{[-2\mu_1-\sigma_1^2]u}\\
&=&e^{-2\sigma_1^2u+2\sigma_1 W_u}
\end{eqnarray*}
hence $Z_u$ is a Girsanov density.  Let $\tilde{\mathbb{Q}}$ be the probability measure such that $\frac{d\tilde{\mathbb{Q}}_u}{d\mathbb{P}_u}=Z_u$. Now, we have that
\begin{eqnarray*}
c(t,v)&=& \sigma_{1}^2 (1-\rho^2)v^2 e^{[-2\mu_1-\sigma_1^2]t} \tilde{\mathbb{Q}}(\int_t^T a(u,V_u)e^{[2\mu_1+\sigma^2]u} (\frac{\partial b}{\partial v}(u,V_u))^2du|V_t=v)
\end{eqnarray*}

Now we substitute the formulas for $a$ and $\frac{\partial b}{\partial v}$:
\begin{eqnarray*}
c(t,v)&=& \sigma_{1}^2 (1-\rho^2)v^2 e^{[-2\mu_1-\sigma^2]t}\times\\
&&\tilde{\mathbb{Q}}(\int_t^T e^{-\left(\frac{\mu_2}{\sigma_2}\right)^2 {T-u}}e^{[2\mu_1+\sigma^2]u} e^{\left(2\mu_1-2\frac{\mu_2\rho\sigma_1}{\sigma_2}\right)(T-u)}N(d_1(u,V_u))^2du|V_t=v)\\
&=&\sigma_{1}^2 (1-\rho^2)v^2 e^{[2\mu_1+\sigma^2](T-t)}\times\\
&&\tilde{\mathbb{Q}}(\int_t^T e^{\left[-\left(\frac{\mu_2}{\sigma_2}\right)^2 -\sigma^2-2\frac{\mu_2\rho\sigma_1}{\sigma_2}\right](T-u)}N(d_1(u,V_u))^2du|V_t=v).\\
\end{eqnarray*}

Because $W_u$ has a drift equal to $2\sigma_1$ with respect to $\tilde{\mathbb{Q}}$, the $\tilde{\mathbb{Q}}$-conditional distribution of $\log V_u=\log V_t+(\mu_1-\frac{1}{2}\sigma_1^2)(u-t))+\sigma_1 (W_u-W_t)$ given $V_t=v$ is normal with mean $\log v+(\mu_1-\frac{\sigma_1^2}{2}+2\sigma_1^2)(u-t)$ and standard deviation $\sigma_1 \sqrt{u-t}$.  Therefore $d_1(u,V_u)$ also has normal distribution, with mean given by formula \eqref{eq:mean} and standard deviation $\frac{\sqrt{u-t}}{\sqrt{T-u}}$.  This proves formula \eqref{formulaforc}.  To complete the proof it is sufficient to show that the expected value of $(N(d_1)^2)$ is monotone decreasing in $v$ where 
\[\textbf{d}= \mu(t,u,v)+ \sqrt{\frac{u-t}{T-u}}Z\]
where $Z$ is a standard normal random variable.  Note that
\begin{eqnarray*}
\frac{\partial \mathbb{E}(N(\textbf{d})^2)}{\partial v}&=&\mathbb{E}\left[\frac{\partial (N(\textbf{d})^2)}{\partial v}\right]\\
&=&\mathbb{E}\left[2 N(\textbf{d})(\frac{1}{\sqrt{2\pi}} e^{-\frac{\textbf{d}^2}{2}})\frac{\partial\mu(u,v)}{\partial v}\right].\\
 \end{eqnarray*}
Since $\frac{\partial\mu(t,u,v)}{\partial v}<0$, and all the other terms inside the expectation are always non-negative, $\frac{\partial \mathbb{E}(N(\textbf{d})^2)}{\partial v} $ is negative as desired.  
\end{proof}

\begin{theorem}\label{theo:bound} Assume $\frac{1}{2}+\frac{\alpha}{\sigma_1^2}>0$.  We have $\frac{M_t}{N_t}=\theta(t,V_t)$ for some function $\theta$ such that $\sup_{t\in[0,T], v>0} \theta(t,v)<\infty$.
\end{theorem}

\begin{proof}
Using the formula \eqref{formulaforM} and the PDEs for $b$ and $c$ we obtain the following formula for $M_t$:
\begin{eqnarray*}
M_t &=&M(t,V_t)
\end{eqnarray*}
where
\begin{eqnarray*}
M(t,v)&=&e^{-\kappa \tilde{c}} \left[\frac{\sigma_1}{\sigma_2}\rho \mu_2 v\frac{\partial b}{\partial v}\right] -\kappa b e^{-\kappa\tilde{c}}a\sigma_1^2(\rho^2-1)(\frac{\partial b}{\partial v})^2+2\kappa b e^{\kappa\tilde{c}}\frac{\partial \tilde{c}}{\partial v} v\sigma_1^2 \\
&&+ (2\mu_1+\sigma_1^2)\kappa b e^{-\kappa\tilde{c}}\tilde{c}\sigma_1^2+\frac{1}{2}\kappa^2 b e^{-\kappa \tilde{c}}(\frac{\partial \tilde{c}}{\partial v})^2\sigma_1^2v^2-\kappa e^{-\kappa\tilde{c}}\frac{\partial \tilde{c}}{\partial v}\frac{\partial b}{\partial v}\sigma_1^2v^2.
\end{eqnarray*} 
We also have $N_t=N(t,v)$ where
\begin{eqnarray*}
N(t,v)= e^{-\kappa \tilde{c}}v\sigma_1 \left[\frac{\partial b}{\partial v}-\kappa b \frac{\partial \tilde{c}}{\partial v}\right].
\end{eqnarray*}
Hence, we have $\frac{M_t}{N_t}=\theta(t,V_t)$ where $\theta(t,v)= \frac{M(t,v)}{N(t,v)}$.  Because $\left[\frac{\partial b}{\partial v}-\kappa b \frac{\partial \tilde{c}}{\partial v}\right]\geq \max(\frac{\partial b}{\partial v}, \kappa b \left|\frac{\partial \tilde{c}}{\partial v}\right|)$, and noting that $a$, $b$, $e^{-\kappa \tilde{c}}$, $\frac{\partial b}{\partial v}$ and $\tilde{c}$ are bounded functions,  we have  $\sup_{t\in[0,T], v>0} |\theta(t,v)|<\infty$ if we show the following:

\begin{enumerate}[label=(\alph*)]
\item $\frac{b}{v}$ is bounded when $t\in[0,T]$, $v>0$,
\item $\frac{\partial \tilde{c}}{\partial v}.v $ is bounded when $t\in[0,T]$, $v>0$,
\item $\frac{b}{\frac{\partial b}{\partial v}v}$ is bounded when $t\in[0,T]$, $v\leq D$, 
\item $\frac{\tilde{c}}{\frac{\partial \tilde{c}}{\partial v}v}$ is bounded when $t\in[0,T]$, $v>D$.
\end{enumerate}

We estimate each quantity as follows:

\begin{enumerate}[label=(\alph*)] 

\item We note that $(\frac{b}{v})' =-\frac{D(1-N(d_2))}{v}<0$ hence the $\sup_v \frac{b}{v}=\lim_{v\rightarrow 0}\frac{b}{v}=\lim_{v\rightarrow 0}\frac{\partial b}{\partial v}=e^{\left(\mu_1-\frac{\mu_2\rho\sigma_1}{\sigma_2}\right)(T-t)}$ which is bounded. 

\item Let $\textbf{d}$ be normally distributed with mean $\mu(u,v)$ and  
with variance $\frac{u-t}{T-u}$. We observe that

\begin{eqnarray*}
\left|\frac{\partial E(N(\textbf{d})^2)}{\partial v}\right|&=&\left|\frac{\partial}{\partial v}\int_{-\infty}^{\infty} \frac{1}{\sqrt{2\pi}}e^{-\frac{z^2}{2}}N(\mu(u,v)+\sqrt{\frac{u-t}{T-u}} z)^2 dz\right|\\
&=&\int_{-\infty}^{\infty} \frac{1}{\sqrt{2\pi}}e^{-\frac{z^2}{2}}2 N(\mu(u,v)+\sqrt{\frac{u-t}{T-u}} z)  e^{-\frac{(\mu(u,v)+\sqrt{\frac{u-t}{T-u}} z)^2}{2}} \frac{1}{v\sigma_1\sqrt{T-u}} dz\\
&\leq& \frac{2}{v \sqrt{2\pi}\sigma_1\sqrt{T-u}}
\end{eqnarray*}
Therefore,
\begin{eqnarray*}
\left|\frac{\partial \tilde{c}}{\partial v}\right|.v&\leq & \sigma_1^2 (1-\rho^2) e^{[(2\mu_1+\sigma_1^2)](T-t)}\int_t^T e^{\left(-(\frac{\mu_2}{\sigma_2})^2-\sigma_1^2-2\frac{\mu_2}{\sigma_2}\rho \sigma_1\right) (T-u)}\frac{2}{\sqrt{2\pi}\sigma_1\sqrt{T-u}}du\\
&\leq & \tilde{K} \int_t^T \frac{1}{\sigma_1\sqrt{T-u}}du\\
&\leq& K \sqrt{T-t}.
\end{eqnarray*}
for some constants $\tilde{K}$ and $K$.  Hence
$\frac{\partial \tilde{c}}{\partial v}.v$  is bounded as desired.

\item Note that
\begin{eqnarray}
\frac{b(t,v)}{\frac{\partial b}{\partial v}(t,v)v}&=&\frac{ve^{k_1(T-t)}N(d_1(t,v))+D(1-N(d_2(t,v))}{ve^{k_1(T-t)}N(d_1(t,v))}\nonumber\\
&=&1+\frac{D (1-N(d_2(t,v))}{vN(d_1(t,v))}\label{eq:novikov1}.
\end{eqnarray}
Next, observe that $\frac{b(t,v)}{\frac{\partial b}{\partial v}(t,v)v}$ is bounded in the domain $[0,T)\times [D/2,D]$, because $d_1(t,v)$ is bounded below by $K=-|\sigma||\frac{1}{2}+\frac{\alpha}{\sigma_1^2}|T$ and
\[\frac{b(t,v)}{\frac{\partial b}{\partial v}(t,v)v}\leq 1+\frac{D}{vN(K)}\leq 1+\frac{2}{N(K)}, \mbox{\, for\, } (t,v)\in [0,T)\times [D/2,D].\]
We claim that $\frac{b(t,v)}{\frac{\partial b}{\partial v}(t,v)v}$ has a continuous extension to the domain $[0,T]\times [0,D/2]$.  Indeed, let$(t_n,v_n)\rightarrow (t,0)$, $t\in[0,T]$. Then $\lim_{n\rightarrow\infty} d_1(t_n,v_n)= \lim_{n\rightarrow\infty} d_2(t_n,v_n)=\infty$.  Hence for $n$ sufficiently large,
\begin{eqnarray*}
\frac{1-N(d_2(t_n,v_n))}{v_n N(d_1(t_n,v_n))} &\leq & \frac{e^{-\frac{(d_2(t_n,v_n))^2}{2}} e^{-\ln(v_n)}}{\sqrt{2\pi}2}.
\end{eqnarray*}
Since $d_2(t_n,v_n))^2$ is a quadratic polynomial in $\ln(v_n)$ with a positive coefficient of $(\ln(v_n))^2$, we have $\lim_{n \rightarrow\infty} -\frac{d_2(t_n,v_n))^2}{2}-\ln(v_n)=-\infty$, hence by equation \eqref{eq:novikov1},  $\lim \frac{b(t_n,v_n)}{\frac{\partial b}{\partial v}(t_n,v_n)v_n}=0$.  Now, let $(t_n,v_n)\rightarrow (T,v)$, $0<v\leq D/2$.  Then $\lim_{n\rightarrow \infty} d_1(t_n,v_n)= \lim_{n\rightarrow \infty} d_2(t_n,v_n)=\infty$. Since $v_nN(d_1(t_n,v_n))\rightarrow v$ and  $D (1-N(d_2(t_n,v_n))\rightarrow 0$, by equation \eqref{eq:novikov1}, $\lim \frac{b(t_n,v_n)}{\frac{\partial b}{\partial v}(t_n,v_n)v_n}=0$ as $(t_n,v_n)\rightarrow (T,v)$ as well. 

Since $\frac{b(t,v)}{\frac{\partial b}{\partial v}(t,v)v}$ can be extended to a continuous function on $[0,T]\times [0,D/2]$, it must also be bounded on $[0,T]\times [0,D/2]$.  We conclude that $\sup_{(t,v)\in [0,T)\times (0,D]}\frac{b(t,v)}{\frac{\partial b}{\partial v}(t,v)v}<\infty$.

\item  Let $K_1= -(\frac{\mu_2}{\sigma_2})^2-\sigma_1^2-2\frac{\mu_2}{\sigma_2}\rho \sigma_1$.  We observe that 
\begin{eqnarray*}
\frac{\tilde{c}(t,v)}{\frac{\partial\tilde{c}}{\partial v}(t,v)v}&=&K_2\frac{\int_t^T e^{K_1(T-u)}\mathbb{E}(N(\textbf{d})^2)du}{\int_t^T\frac{e^{K_1(T-u)}}{\sigma_1\sqrt{T-u}}\mathbb{E}[N(\textbf{d})e^{-\frac{d_1^2}{2}}]du},
\end{eqnarray*}
where $K_2>0$ is a constant independent of $t$ and $v$.  We define 
\[z(u,v)=  \frac{-\mu(u,v)\sqrt{T-u}}{\sqrt{u-t}},\]
where $\mu(u,v)$ is defined by equation \eqref{eq:mean}.  Note that if we write $\textbf{d}= \frac{\sqrt{u-t}}{\sqrt{T-u}} Z+\mu(u,v)$, then we have $\textbf{d}>0$ if and only if $Z>z(u,v)$.  Next, we decompose $\int_t^T e^{K_1(T-u)}\mathbb{E}(N(\textbf{d})^2)du$ into 3 parts:
\begin{eqnarray*}
\int_t^T e^{K_1(T-u)}\mathbb{E}(N(\textbf{d})^2)du&=&\int_t^T e^{K_1(T-u)}\mathbb{E}(N(\textbf{d})^2 1_{\{\textbf{d}<-1\}})du\\
&& + \int_t^T e^{K_1(T-u)}\mathbb{E}(N(\textbf{d})^2 1_{\{\textbf{d}\geq -1\}})1_{\{z_{u,v}\leq 1\}}du\\
&& + \int_t^T e^{K_1(T-u)}\mathbb{E}(N(\textbf{d})^2 1_{\{\textbf{d}\geq -1\}})1_{\{z_{u,v}> 1\}}du
\end{eqnarray*}
Let $I$, $II$, and $III$ be the integrals on the right side of the equation above.

We have that $N(\textbf{d})^2 1_{\{\textbf{d}<-1\}}\leq N(\textbf{d})\frac{1}{\sqrt{2\pi}}e^{-\frac{\textbf{d}^2}{2}} 1_{\{\textbf{d}<-1\}}$, hence $\mathbb{E}(N(\textbf{d})^2 1_{\{\textbf{d}<-1\}})\leq \frac{1}{\sqrt{2\pi}}\mathbb{E}(N(\textbf{d})e^{-\frac{\textbf{d}^2}{2}})\leq K_3 \frac{\mathbb{E}(N(\textbf{d})e^{-\frac{\textbf{d}^2}{2}})}{\sigma_1\sqrt{T-u}}$ where $K_3=K_2 \frac{1}{\sqrt{2\pi}}\sigma_1\sqrt{T}$.  Hence

\begin{eqnarray*}
\frac{I}{\int_t^T\frac{e^{K_1(T-u)}}{\sigma_1\sqrt{T-u}}\mathbb{E}[N(\textbf{d})e^{-\frac{\textbf{d}^2}{2}}]du}\leq K_3.
\end{eqnarray*}

Next, we argue that $z(u,v)\leq 1$ implies that $\mathbb{E}(N(\textbf{d})e^{-\frac{\textbf{d}^2}{2}})> \epsilon \min(1, \frac{\sqrt {T-u}}{\sqrt{t-u}})$ where $\epsilon$ is a positive number independent of $t$ and $v$.  We first observe that
\[z(u,v)=-\frac{\ln(D/v)}{\sigma_1\sqrt{u-t}}+\frac{\tilde{\alpha}\sqrt{u-t}}{\sigma_1}+\sigma_1(\frac{1}{2}+\frac{\alpha}{\sigma_1^2})\frac{T-u}{\sqrt{u-t}},
\]
where $\tilde{\alpha}=\mu_1+\frac{3}{2}\sigma_1^2$.
Because of the assumptions $v>D$ and  $\frac{1}{2}+\frac{\alpha}{\sigma_1^2} >0$,  $z(u,v)> \frac{-|\tilde{\alpha}|\sqrt{T}}{\sigma_1}$ for all $u,v$.  Now,  $\mathbb{E}(N(\textbf{d})e^{-\frac{\textbf{d}^2}{2}})\geq K_4 \mathbb{P}(z(u,v)-\sqrt{\frac{T-u}{t-u}}<Z< z(u,v))$ for some constant $K_4$,  because $z(u,v)-\sqrt{\frac{T-u}{t-u}}\leq Z\leq  z(u,v)$ implies that $-1\leq \textbf{d}\leq 0$.

If $\sqrt{\frac{T-u}{t-u}}>1$, then   $\mathbb{P}(z(u,v)-\sqrt{\frac{T-u}{t-u}}<Z< z(u,v))\geq e^{-\frac{1}{2}\max(z(u,v)^2,(z(u,v)-1)^2)}$, which is bounded below by $\epsilon>0$ independent of $u$,$v$ or $t$, since $1\geq z(u,v)>-  \frac{\tilde{\alpha}\sqrt{T}}{\sigma_1}$.  If $\sqrt{\frac{T-t}{t-u}}\leq 1$, then $\mathbb{P}(z(u,v)-\sqrt{\frac{T-u}{t-u}}<Z< z(u,v))\geq \sqrt{\frac{T-u}{t-u}}e^{-\frac{1}{2}\max(z(u,v)^2,(z(u,v)-1)^2)}$, which proves the claim.  

Now, we note that $\mathbb{E}(N(\textbf{d})^2 1_{\{\textbf{d}\geq -1\}})\leq \mathbb{P}(Z> z(u,v)-\sqrt{\frac{T-u}{t-u}})$
Hence,
\begin{eqnarray*}
\frac{II}{\int_t^T \frac{e^{K_1(T-u)}}{\sigma_1\sqrt{T-u}}\mathbb{E}[N(\textbf{d})e^{-\frac{\textbf{d}^2}{2}}]du}&\leq&K_5 +\frac{\int_t^T \mathbb{P}(Z>z(u,v))1_{\{z(u,v)\leq 1\}} du}{\epsilon\int_t^T \frac{\min(1, \frac{\sqrt {T-u}}{\sqrt{t-u}})}{\sigma_1\sqrt{T-u}}1_{\{z(u,v)\leq 1\}}du}\\
&\leq& K_5+K_6\frac{\int_t^T e^{K_1(T-u)} 1_{\{z(u,v)\leq 1\}} du}{\int_t^T e^{K_1(T-u)} 1_{\{z(u,v)\leq 1\}}du}\\
&=& K_5+K_6,
\end{eqnarray*}
where $K_6=\frac{\sigma_1 \sqrt{T}}{\epsilon}$.

Finally, we observe that
\begin{eqnarray*}
\frac{III}{\int_t^T \frac{e^{K_1(T-u)}}{\sigma_1\sqrt{T-u}}\mathbb{E}[N(\textbf{d})e^{-\frac{\textbf{d}^2}{2}}]du}&\leq& K_7 +\frac{\int_t^T \mathbb{P}(Z>z(u,v))1_{\{z(u,v)>1\}} du}{\int_t^Te^{-\frac{1}{2}\max(z(u,v)^2,(z(u,v)-1)^2)}  \frac{\min(1, \frac{\sqrt {T-u}}{\sqrt{t-u}})}{\sigma_1\sqrt{T-u}}1_{\{z(u,v)>1\}}du}\\
&\leq &K_7+ \frac{\int_t^T K_8 e^{-\frac{z(u,v)^2}{2}}1_{\{z(u,v)>1\}} du}{\int_t^T e^{-\frac{1}{2}z(u,v)^2} \frac{\min(1, \frac{\sqrt {T-u}}{\sqrt{t-u}})}{\sigma_1\sqrt{T-u}}1_{\{z(u,v)>1\}}du}.
\end{eqnarray*}
Hence,
\begin{eqnarray*}
\frac{III}{\int_t^T \frac{e^{K_1(T-u)}}{\sigma_1\sqrt{T-u}}\mathbb{E}[N(\textbf{d})e^{-\frac{\textbf{d}^2}{2}}]}\leq K_7 +K_8\sigma_1 \sqrt{T} \frac{\int_t^T e^{-\frac{1}{2}z(u,v)^2} 1_{\{z(u,v)>1\}} du}{\int_t^T e^{-\frac{1}{2}z(u,v)^2} 1_{\{z(u,v)>1\}}du}\\
&\leq &  K_7 +\frac{K_8}{\sigma_1 \sqrt{T}}
\end{eqnarray*}

%We note that 

%E(N(\mu(u,v)+\sqrt{\frac{u-t}{T-u}} z)  e^{-\frac{(\mu(u,v)+\sqrt{\frac{u-t}{T-u}} z)^2}{2}\geq 
%$E(N(d_1)^2)\leq E(e^{-
%\[\theta(t,v) \leq K_1 + K_2 a(t) \frac{b(t,v)}{v}\frac{\partial b}{\partial v}(t,v)+K_3 \left|\frac{\partial \tilde{c}}{\partial v}(t,v)\right|v+K_4 \frac{b\tilde{c}}{\frac{\partial \tilde{c}}{\partial v}(t,v)v}\]

\end{enumerate}

\end{proof}

\begin{corollary} Assume that $\frac{1}{2}+\frac{\alpha}{\sigma_1^2}>0$. Let $\tilde{c}(t,v)=c(t,v)/v^2$. For any $\kappa>0$, $B_t= b(t,V_t)e^{-\kappa \tilde{c}(t,V_t)}$ gives an arbitrage free price for the Merton style bond $\min(V_T,D)$ in the sense of NFLVR.
\end{corollary}

\begin{proof}  Because $\frac{\partial b}{\partial v}>0$ and $\frac{\partial\tilde{c}}{\partial v}<0$,  $\frac{\partial b}{\partial v}-\kappa \frac{\partial\tilde{c}}{\partial v}>0$.  Also Theorem \eqref{theo:bound} implies that hypothesis \eqref{novikov} of Theorem \eqref{noarbitrage} is satisfied.  Hence by Theorem \eqref{noarbitrage} there exists a probability measure $\mathbb{Q}$ equivalent to $\mathbb{P}$ such that $S_t$ and $B_t$ are $\mathbb{Q}$-local martingales, hence $(S_t, B_t)$ do not allow arbitrage.
\end{proof} 

\begin{remark}  It is possible that no arbitrage property still holds without the assumption $\frac{1}{2}+\frac{\alpha}{\sigma_1^2}>0$, however one may need to use a different criterion other than the Novikov criterion.  A possible approach is to use the absolute continuity of the laws processes of the diffusion type (see e.g. Chapter 7 of \cite{LipShir}). This would require the formulation of $(B_t, S_t)$ as a solution of a stochastic differential equation in the sense of Lipster and Shiryaev (in particular, see section 7.6.4 of \cite{LipShir}).  To keep the presentation simple, we will not pursue this approach in this paper.    
 
\end{remark}
\section{The effects of the parameters on the price} \label{sec:numerical}

In this section, we fix $T=10$, $D=100$ and look at how various parameters affect the price of the Merton style bond.  Examining the formulas we see that there are five  parameters affecting the price:   $\mu_1$, $\theta:=\frac{\mu_2}{\sigma_2}$,  $\rho$, $\sigma_1$ and $\kappa$. 

\subsection{The effect of $\mu_1$}

We observe that $\mu_1$ affect $b$ only through the parameter $\alpha= \mu_1-\frac{\mu_2\rho\sigma_1}{\sigma_2}$.  Hence we can calculate $\frac{\partial b}{\partial \alpha}$ and deduce $\frac{\partial b}{\partial \mu_{1}}$ from this.  An elementary calculation gives
\[\frac{\partial b}{\partial \alpha}(t,v)=v(T-t)e^{\alpha(T-t)}N(d_1),\]
which is always positive.  Since $\frac{\partial\alpha}{\partial\mu_1}=1$, we find that $b$ is strictly increasing with $\mu_1$. Note that in our pricing model, we may interpret $b$ as the price of the Merton style bond when $\kappa=0$.  Hence 
\[y(t)=-\log \frac{D}{b(t,V_t)}\]
is the yield of the bond for the case $\kappa=0$. The fact that $b$ is increasing with $\mu_1$ means that the growth rate of the underlying firm, $\mu_1$, affects the yield of the bond negatively which may be because the bond becomes less risky.

We also observe that when $\alpha=0$, $b$ is the price of the bond in the Merton model.  Hence, in this model, depending on the sign of $\alpha$, the price of the bond can be higher (when $\alpha>0$)  or lower ($\alpha<0$) than its price in the Merton model.  

The relationship between $\mu_1$ and $\tilde{c}$ is not monotone, even when $\alpha$ remains constant.  
%\begin{figure}
% \centering
 % \includegraphics[width=\textwidth]{mu1.pdf}
 % \caption{The change in $\tilde{c}$ when $\mu_1$ is increased from $0$ to $0.2$. $\mu_2$ = 0.2, $\sigma_1 = 0.2$, $\sigma_2 = 0.3$, $\rho = 0.6$}
%\label{mu1}  
%\end{figure}

\subsection{The effect of $\theta$}

This parameter affects $b$ only through $\alpha$ as well. Since $\frac{\partial\alpha}{\partial\theta}=-\rho \sigma_1$, $b$ is decreasing with $\theta$ if $\rho>0$ and increasing with $\theta$ if $\rho<0$.  If $\rho=0$, $\theta$ has no effect on $b$.    One can interpret $\theta$ as a measure of systemic risk in the underlying security.  $\rho>0$ implies that the optimal replicating portfolio is also positively correlated with the same systemic risk, therefore it is intuitive that this portfolio gives a premium that increases with $\rho$.  Also, the greater $\sigma_1$, the greater the exposure to the systemic risk, hence it is again intuitive that the premium increases with $\sigma_1$ as well. 

$\theta$ has a monotone decreasing relationship with $\tilde{c}$ provided that $\alpha$ remains constant.  But since $\alpha$ is a function of $\theta$, in general there is no monotone relationship between $\theta$ and $\tilde{c}$.   %We calculated $\tilde{c}$ as $\theta$ varies for various scenarios.   In Figures \eqref{theta133+} and \eqref{theta66+}, we plotted $\tilde{c}$ as $\theta$ varies from $0$ to $0.8$ setting $\mu_1$ = 0, $\sigma_1 = 0.2$, $\sigma_2 = 0.3$, and $\rho=0.6$, for both cases $v=66.66$ and $v=133.33$.   In Figures \eqref{theta133-} and \eqref{theta66-},  we repeat the same experiment but this time with $\rho=-0.6$.  

%\begin{figure}
 %\centering
  %\includegraphics[width=\textwidth]{theta.pdf}
  %\caption{The change in $\tilde{c}$ when $\theta$ is increased from $0$ to $0.66$. $\mu_1= 0.2$, $\sigma_1 = 0.2$, $\rho = 0.6$}
%\label{theta+}  
%\end{figure}    
    
%\begin{figure}
 %\centering
  %\includegraphics[width=\textwidth]{theta-.pdf}
  %\caption{The change in $\tilde{c}$ when $\theta$ is increased from $0$ to $0.66$. $\mu_1= 0.2$, $\sigma_1 = 0.2$, $\rho = -0.6$}
%\label{theta-}  
%\end{figure}    

\subsection{The effect of $\sigma_1$} 

\begin{figure}[h]
 \centering
  \includegraphics[height=4.5in,width=6.5in]{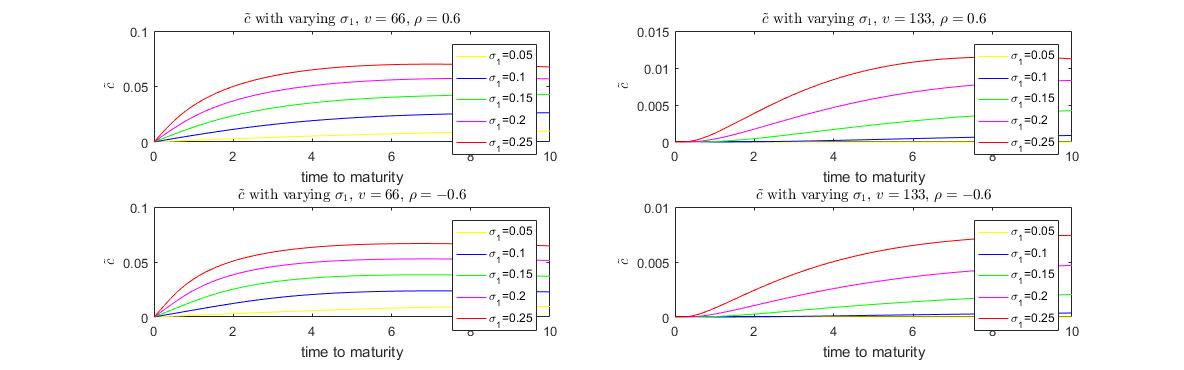}
  \caption{$\tilde{c}$ with varying $\sigma_1$, $\mu_1$ = 0.02,  $\theta=0.4$}
\label{sigmaall}  
\end{figure}    
We explicitly calculate  $\frac{\partial b}{\partial \sigma_1}$ as
\[ \frac{\partial b}{\partial \sigma_1}= v(T-t)e^{\alpha(T-t)}N(d_1) \frac{\partial{\alpha}}{\partial{\sigma_1}}- D\frac{e^{-\frac{d_2^2}{2}}}{\sqrt{2\pi}} \sqrt{T-t}\]
Note if $\rho\mu_2\geq 0$, then  $\frac{\partial{\alpha}}{\partial{\sigma_1}}<0$, hence, in this case $b$ is monotone decreasing as $\sigma_1$ increases.  If $\rho\mu_2<0$, it is not clear if the effect of $\sigma_1$ on $b$ would be monotone.  

In Figure \eqref{sigmaall} we plotted $\tilde{c}$ as $\sigma_1$ varies from $0.05$ to $0.25$ for each of the cases  $\rho=0.6$, $v=66$ and $\rho=0.6$, $v=133$, and $\rho=-0.6$, $v=66$ and $\rho=-0.6$, $v=133$, while setting $\mu_1$ = 0.02, $\theta=0.4$ .   We see that in all the cases $\tilde{c}$ increases as $\sigma_1$ increases, suggesting that there is a monotone increasing relationship between $\tilde{c}$ and $\sigma_1$.  However, we have not confirmed this theoretically.

\subsection{The effect of $\rho$}

The correlation parameter $\rho$ affects both the price of the replicating portfolio and the replication error.  Examining the formula for $b$, we see that $\rho$ appears only in the term $\alpha$.  Since $\alpha$ is monotone in $\rho$ and $b$ is monotone in $\alpha$, $b$ is monotone in $\rho$.  In particular, $b$ is monotone decreasing (resp. increasing) as $\rho$ increases  if $\mu_2>0$ (resp. $\mu_2<0)$.  If $\mu_2=0$ then $b$ is constant in $\rho$.

The replication error $c$ is a product of $1-\rho^2$ and another term which also depends on $\rho$, but only through $\alpha$. Hence if $\alpha$ remains constant, $\tilde{c}$ will decrease as $\rho^2$ increases, which is intuitive.  But, surprisingly, there is no monotone relationship between $\tilde{c}$ and $\rho^2$ in general.   In Figure \eqref{rho66133}, we plotted $\tilde{c}$ as $\rho$ varies from $-0.6$ to $0.6$ setting $\mu_1$ = 0.02, $\sigma_1=0.15$, $\theta=0.4$ for each of the cases $v=66$ and $v=132$.   For example, when $v=132$, $\tilde{c}$ does not decrease when $\rho$ is changed from $0$ to $0.3$.  We also observe that the replication error is smaller for negative values of $\rho$ when $\rho^2$ is the same.
\begin{figure}[h]
 \centering
  \includegraphics[height=3in,width=7in]{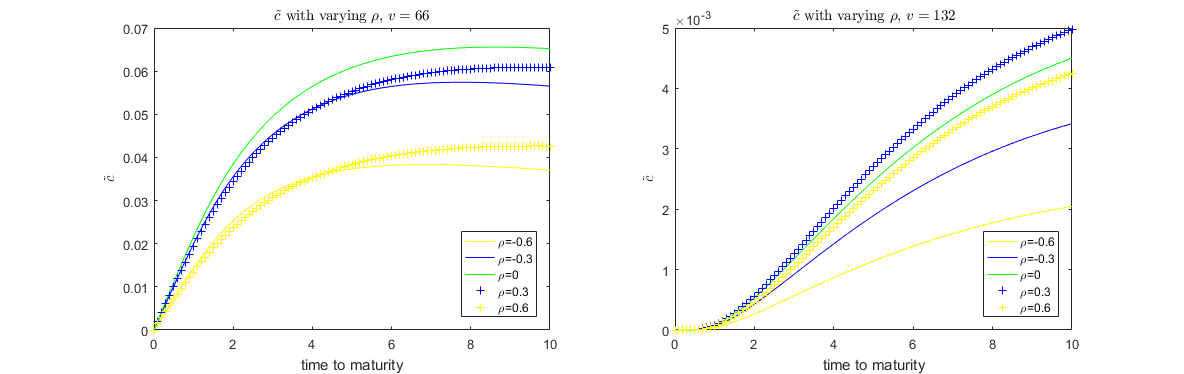}
  \caption{$\tilde{c}$ with varying $\rho$, $\mu_1$ = 0.02, $\sigma_1=0.15$, $\theta=0.4$}
\label{rho66133}  
\end{figure}

\subsection{The yield spreads and the effect of $\kappa$}

\begin{figure}[p]
 \centering
  \includegraphics[height=8in,width=7in]{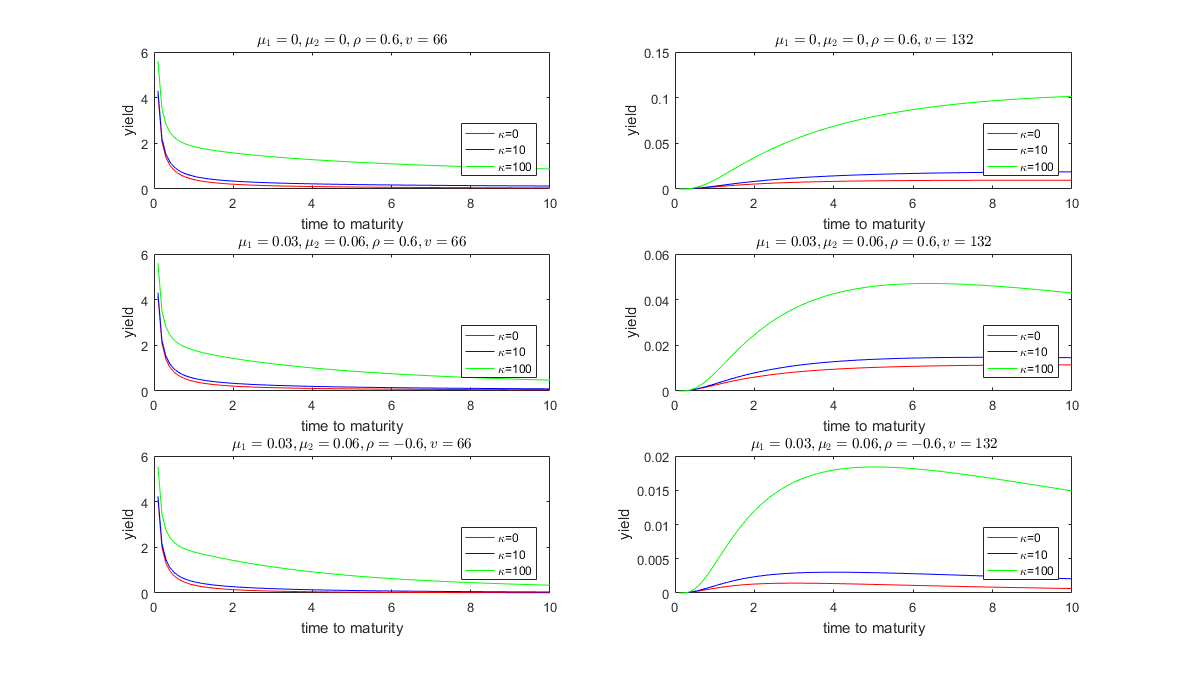}
  \caption{yield with varying $\kappa$ for various configurations}
\label{yieldall}  
\end{figure}

The yield at date $t$ of a Merton style bond with maturity $T$ is defined as the function 
\[y(t,T)= \frac{\log(D)-\log(B_t)}{T-t}\] 
Since we assume that $B_t$ is already discounted, $y(t,T)$ is also the yield spread. Recall, given the parameters, the price of the bond at time $t$ is a function of $T-t$ and the firm value $V_t$ only, hence we can represent $y(t,T)=y(V_t,T-t)$.  In this section, we will plot the function $y(v,T)$ as a function of $T$ (time to maturity), for $v=66$ and $v=132$.  $y(v,T)$ can be interpreted as the time $0$ yield of a Merton style bond with maturity $T$, when the firm value at time $0$ is equal to $v$.  $v=66$ and $v=132$ represent two distinct scenarios where the firm's value is respectively less and greater than the face value of the debt ($D=100$).        

In the first set of plots in Figure \eqref{yieldall} we set $\sigma_1=0.15$, $\rho=0.6$, and $\mu_1=\mu_2=0$.  Note that this gives $\alpha=\mu_1-\rho\theta \sigma_1=0$.  Since the price of the replicating portfolio is the Merton model's price when $\alpha=0$, we obtain the same yield curves as the Merton model when $\kappa=0$.  The yield curve is monotone decreasing when the firm's value is $66$ and monotone increasing when $v=133$  As we vary $\kappa$ the yield spreads increase, however the shape of the curves remain the same.

In the second set of plots in Figure \eqref{yieldall} we set $\mu_1=0.03$ and $\theta=0.4$, while keeping $\sigma_1=0.15$ and $\rho=0.6$ as before. Note that in this configuration $\alpha=-0.006$, which causes an increase in the yield spread at the baseline case, $\kappa=0$.  This is because the price of the replicating portfolio increases with $\alpha$.  The parameter $\kappa$ increases the yield spreads, and when $v=132$, and we see a hump shape pattern emerging when $\kappa$ increases to $100$.  

In the third set of plots in Figure \eqref{yieldall}, we set $\mu_1=0.03$ and $\theta=0.4$, $\sigma_1=0.15$ and $\rho=-0.6$.  In this configuration note that $\alpha=0.009$,  hence the yields decrease in the baseline case, $\kappa=0$.  The yields increase with $\kappa$ as before, and the hump shape of the yield spread is more noticeable when $\kappa=100$ and $v=132$.

\section{Conclusion}  
To summarize our results, we proposed a pricing formula for a contingent claim on the assets of a firm where the assets are not traded.  The price has two components, one is the price of the mean variance portfolio constructed by a correlated asset that is traded in the market.  The second component is a discount factor which is determined by a parameter $\kappa$ and the error of the replication.  In the specific application to the pricing of the Merton style bond, we are able to show that the pricing formula is arbitrage free without any restrictions on the value of $\kappa$.  This gives a wide range of prices for the Merton style bond, and the value of $\kappa$ can be interpreted as the representative investor’s risk aversion towards the non-hedgeable risk.

We are able to describe meaningful relationships between most parameters and the price of the mean variance hedging portfolio via both theoretical and numerical means.   We are not able to qualitatively describe the relationship between the parameters and the replication error by theoretical means.  However, numerical simulations suggest that $\sigma_1$ has a monotone relationship with the replication error.  The numerical simulations indicated non-monotone relationships between the replication error and the other parameters.

\section{Appendix}
As the boundary condition for $a$ is independent of $s$ and $v$, it is natural to guess that the solution to the PDE (\ref{eq_a}) is independent of $s$ and $v$. It is easy to verify that the function
\begin{equation}\label{sol_a}
a(t) = e^{\left(-\frac{\mu_2}{\sigma_2}\right)^2 t}
\end{equation}
is a solution to the PDE (\ref{eq_a}). Moreover, (\ref{sol_a}) is a unique solution according to Theorem 8.1 in \cite{Lad}. Therefore, we obtain that the unique solution to PDE (\ref{eq_a}) is
\begin{equation}
\begin{cases}
a(t) = e^{-\left(\frac{\mu_2}{\sigma_2}\right)^2 {T-t}},\\
a(T) = 1.
\end{cases}
\end{equation}

To treat the PDE (\ref{eq_b}), we use the change of variables:
\begin{align}
\tau&=T-t, \qquad \tau \in [0, T];\label{change1}\\
u&=\ln v, \qquad u \in (-\infty, \infty);\label{change2}\\
w&=\ln s, \qquad w \in (-\infty, \infty).\label{change3}
\end{align}  
Then the PDE (\ref{eq_b}) becomes
 \begin{align}\label{sub_b}
\frac{\partial b}{\partial \tau} =&  \frac{1}{2}\sigma_1^2\frac{\partial^2 b}{\partial u^2}+\frac{1}{2}\sigma_2^2\frac{\partial^2 b}{\partial w^2}
+\sigma_1\sigma_2\rho\frac{\partial^2 b}{\partial u\partial w}-\frac{1}{2}\sigma_1^2\frac{\partial b}{\partial u}-\frac{1}{2}\sigma_2^2\frac{\partial b}{\partial w}\nonumber\\
&-\left(\frac{\sigma_1\rho \mu_2}{\sigma_2}-\mu_1\right)\frac{\partial b}{\partial u}+\frac{\sigma_1^2}{a}(\rho^2-1)\frac{\partial a}{\partial u}\frac{\partial b}{\partial u},
\end{align}
with initial condition at $\tau$ = 0:
\begin{equation*}
b(0,w,u) = \min(e^u,D)
\end{equation*}

Note that the initial condition as well as the coefficients in equation (\ref{sub_b}) are independent of $s$. We thus may assume that the solution to (\ref{sub_b}) is independent of $s$. In order to put the equation into a canonical form, we can make the following transformation
\begin{equation}\label{tilde_b}
b(\tau, u) = \tilde{b}(\tau, u)e^{\eta u+\beta \tau},
\end{equation}
where $\eta$ and $\beta$ are chosen so that the lower order terms disappear. One can check that by choosing
$$
\eta=\frac{1}{2}+\frac{1}{\sigma_1^2}\left[\frac{\sigma_1\rho \mu_2}{\sigma_2}-\mu_1\right],
$$
$$
\beta=-\left[\frac{\sigma_1^2}{2}+\left(\frac{\sigma_1\rho \mu_2}{\sigma_2}-\mu_1\right)\right]\eta+\frac{\sigma_1^2}{2}\eta^2= -\frac{\sigma_1^2}{2}\eta^2,
$$
the PDE (\ref{sub_b}) reduces to
\begin{equation*}
\frac{\partial \tilde{b}}{\partial \tau} = \frac{1}{2}\sigma_1^2\frac{\partial^2 \tilde{b}}{\partial u^2},
\end{equation*}
with initial condition
\begin{equation*}
\tilde{b}(0,u) = e^{-\eta u}\min(e^u,D).
\end{equation*}

It is known \cite{Evans} that the solution to the above Cauchy problem is unique and is given by
\begin{equation*}
\tilde{b}(\tau,u)=\frac{1}{\sigma_1\sqrt{2\pi\tau}}\int\limits_{-\infty}^{+\infty}\tilde{b}(0,y)e^{-\frac{(u-y)^2}{2\sigma_1^2\tau}}dy.
\end{equation*}
Substituting $\tilde{b}(0,y) = e^{-\eta y}\min(e^y,D)$, we obtain
\begin{eqnarray*}
\tilde{b}(\tau,u)& = &\frac{1}{\sigma_1\sqrt{2\pi\tau}}\int\limits_{-\infty}^{\ln D}e^{(1-\eta)y}e^{-\frac{(u-y)^2}{2\sigma_1^2\tau}}dy+\frac{1}{\sigma_1\sqrt{2\pi\tau}}\int\limits_{\ln D}^{+\infty}De^{-\eta y} e^{-\frac{(u-y)^2}{2\sigma_1^2\tau}}dy \label{sol_tilde_b}\\
& = & e^{u(1-\eta)+\frac{\sigma_1^2\tau(1-\eta)^2}{2}}N\left(\frac{\ln D-u-\sigma_1^2\tau(1-\eta)}{\sigma_1\sqrt{\tau}}\right)-D e^{-u\eta+\frac{\sigma_1^2\tau\eta^2}{2}}N\left(\frac{\ln D-u-\sigma_1^2\tau\eta}{\sigma_1\sqrt{\tau}}\right)\nonumber\\
&& +\: D e^{-u\eta+\frac{\sigma_1^2\tau\eta^2}{2}}\label{Sol_tilde_b}.
\end{eqnarray*}
Substituting (\ref{change1}), (\ref{change2}) and (\ref{tilde_b}), we obtain the unique solution to the PDE (\ref{eq_b}):
\begin{equation}\label{sol_b}
\begin{cases}
b(t, v) =  ve^{\left(\mu_1-\frac{\mu_2\rho\sigma_1}{\sigma_2}\right)(T-t)}N(d_1)+D\left(1-N(d_2)\right),\\
b(T, v) = \min(v, D),
\end{cases}
\end{equation}
where $N$ is the standard normal distribution function and
\begin{eqnarray*}
&& d_1 = \frac{\ln \frac{D}{v}-\sigma_1^2(1-\eta)(T-t)}{\sigma_1\sqrt{T-t}},\\
&& d_2 = d_1+\sigma_1\sqrt{T-t},
\end{eqnarray*}
where we have $\eta=\frac{1}{2}+\frac{1}{\sigma_1^2}\left[\frac{\sigma_1\rho \mu_2}{\sigma_2}-\mu_1\right]$.  Note that $1-\eta=\frac{1}{2}+\frac{\alpha}{\sigma_1^2}$ where $\alpha=\mu_1-\frac{\mu_2\rho\sigma_1}{\sigma_2}$.  
In regarding to the PDE for $c$, by the change of variables (\ref{change1}), (\ref{change2}) and (\ref{change3}), the equation (\ref{eq_c}) becomes
\begin{align}\label{change_c}
\frac{\partial c}{\partial \tau} = & \mu_{1}\frac{\partial c}{\partial u}+\mu_{2}\frac{\partial c}{\partial w}+\frac{1}{2}\sigma_{1}^2(\frac{\partial^2c}{\partial u^2}-\frac{\partial c}{\partial u})+\frac{1}{2}\sigma_{2}^2(\frac{\partial^2c}{\partial w}-\frac{\partial c}{\partial w})\nonumber\\
&- a\sigma_{1}^2(\rho^2-1)(\frac{\partial b}{\partial u})^2,
\end{align}
with initial condition at $\tau = 0$:
\begin{equation*}
c(0, w, u) = 0.
\end{equation*}

Note that the initial condition as well as the coefficients in equation (\ref{change_c}) are independent of $s$. We hence may assume that the solution to (\ref{change_c}) is independent of $w$. Similar to what we have done for $b$, we make the following transformation
\begin{equation}\label{tilde_c}
c(\tau, u) = \tilde{c}(\tau, u)e^{\alpha^1 u+\beta^1 \tau},
\end{equation}
with
$$
\alpha^1=\frac{1}{2}-\frac{1}{\sigma_1^2}\mu_1,
$$
$$
\beta^1=\left(\mu_1-\frac{1}{2}\sigma_1^2\right)\alpha^1+\frac{\sigma_1^2}{2}(\alpha^1)^2 = -\frac{\sigma_1^2}{2}(\alpha^1)^2.
$$

Therefore, the PDE (\ref{change_c}) reduces to
\begin{equation}\label{red_c}
\frac{\partial \tilde{c}}{\partial \tau} = \frac{1}{2}\sigma_1^2\frac{\partial^2 \tilde{c}}{\partial u^2}+e^{-\alpha^1u-\beta^1\tau}a\sigma_1^2(1-\rho^2)(\frac{\partial b}{\partial u})^2,
\end{equation}
with initial condition
\begin{equation*}
\tilde{c}(0, u) = 0.
\end{equation*}

In order to solve the above PDE, we need the value of $\frac{\partial b}{\partial u}$. Firstly, we differentiate the expression (\ref{sol_tilde_b}) to obtain
\begin{equation*}
\frac{\partial \tilde{b}}{\partial u} = \frac{1}{\sigma_1\sqrt{2\pi\tau}}\left(\int\limits_{-\infty}^{\ln D}\frac{-2(u-y)}{2\sigma_1^2\tau}e^{(1-\eta)y}e^{-\frac{(u-y)^2}{2\sigma_1^2\tau}}dy+\int\limits_{\ln D}^{+\infty}\frac{-2(u-y)}{2\sigma_1^2\tau}De^{-\eta y} e^{-\frac{(u-y)^2}{2\sigma_1^2\tau}}dy\right).
\end{equation*}
Integrating by parts we have
\begin{equation}\label{der_tilde_b}
\frac{\partial \tilde{b}}{\partial u}
=\frac{1}{\sigma_1\sqrt{2\pi\tau}}\left(\int\limits_{-\infty}^{\ln D}(1-\eta)e^{(1-\alpha)y}e^{-\frac{(u-y)^2}{2\sigma_1^2\tau}}dy
-D\eta\int\limits_{\ln D}^{+\infty}e^{-\eta y} e^{-\frac{(u-y)^2}{2\sigma_1^2\tau}}dy\right).
\end{equation}
Using the definition (\ref{tilde_b}), it is easy to check that
$$
\frac{\partial b}{\partial u}=\left(\frac{\partial \tilde{b}}{\partial u}+\eta \tilde{b}\right)e^{\eta u+\beta \tau}.
$$
Thus,
\begin{equation}\label{der_b}
\begin{split}
\frac{\partial b}{\partial u}=& \frac{e^{\eta u+\beta \tau}}{\sigma_1\sqrt{2\pi\tau}}
\int\limits_{-\infty}^{\ln D}e^{(1-\eta)y}e^{-\frac{(u-y)^2}{2\sigma_1^2\tau}}dy\\
=& e^{\frac{\sigma_1^2\tau(1-\eta)^2}{2}+u+\beta\tau}N\left(\frac{\ln D -\sigma_1^2\tau(1-\eta)-u}{\sigma_1\sqrt{\tau}}\right)\\
=& e^{u+\left(\mu_1-\frac{\mu_2\rho\sigma_1}{\sigma_2}\right)\tau}N(d_1).
\end{split}
\end{equation}
By substituting (\ref{der_b}) into (\ref{red_c}), we are able to obtain the unique solution to the PDE (\ref{red_c}) in the form
\begin{equation}
\begin{cases}
\tilde{c}(\tau, u) = \int_0^\tau \frac{1}{\sigma_1\sqrt{2\pi(\tau-s)}}\int_{-\infty}^{\infty}e^{-\frac{y^2}{2\sigma_1^2(\tau-s)}}f(s, u-y)dy ds,\\
\tilde{c}(0, u) = 0,
\end{cases}
\end{equation}
where we have the substitutions
\begin{equation*}
\begin{cases}
u = \ln v,\\
f(\tau, u) = e^{-\alpha^1u-\beta^1\tau}a\sigma_1^2(1-\rho^2)(\frac{\partial b}{\partial u})^2,\\
c(\tau, u) = \tilde{c}(\tau, u)e^{\alpha^1 u+\beta^1 \tau},\\
\alpha^1=\frac{1}{2}-\frac{1}{\sigma_1^2}\mu_1,\\
\beta^1=-\frac{\sigma_1^2}{2}(\alpha^1)^2.
\end{cases}
\end{equation*}

\bibliography{Financebib}
\bibliographystyle{alpha}
\end{document}